\documentclass[
final
]{dmtcs-episciences}

\usepackage[utf8]{inputenc}
\usepackage{subfigure}

\usepackage{amssymb,amsmath}
\usepackage[normalem]{ulem}

\newtheorem{definition}{Definition}
\newtheorem{lemma}{Lemma}

\newtheorem{theorem}{Theorem}
\newtheorem{corollary}{Corollary}

\newcommand{\cut}[1]{{}}
\newcommand{\B}{\mathcal{B}}
\newcommand{\seg}[1]{\text{\sout{\ensuremath{\:{#1}\:}}}}

\author[Ackerman et al.]{
   Eyal~Ackerman\affiliationmark{1}%\thanks{Email: \email{ackerman@sci.haifa.ac.il}.}
\and
   Michelle~M.~Allen\affiliationmark{2}%\thanks{Email: \email{mallen@cs.tufts.edu}.}
\and
   Gill~Barequet\affiliationmark{3}%\thanks{Email: \email{barequet@cs.technion.ac.il}.}
\and
   Maarten~L\"{o}ffler\affiliationmark{4}\\%\thanks{Email: \email{m.loffler@uu.nl}.}
\and
   Joshua~Mermelstein\affiliationmark{2}%\thanks{Email: \email{joshua.mermelstein@tufts.edu}.}
\and
   Diane~L.~Souvaine\affiliationmark{2}%\thanks{Email: \email{dls@cs.tufts.edu}.}
\and
   Csaba~D.~T\'{o}th\affiliationmark{2,5}%\thanks{Email: \email{cdtoth@acm.org}.}
}

\title{The Flip Diameter of Rectangulations and Convex Subdivisions\thanks{A preliminary version of this work
appeared in the Proceedings of the 11th Latin American Theoretical INformatics Symposium (LATIN 2014), LNCS 8392, Springer, 2014, pp. 478--489.
Research on this paper was partially supported by the NSF grant CCF-0830734
and Netherlands Organization for Scientific Research (NWO) under grant~639.021.123.}}
\affiliation{
  Department of Mathematics, Physics, and Computer Science,
     University of Haifa at Oranim, Tivon, Israel \\
  Department of Computer Science, Tufts University, Medford, MA, USA \\
  Department of Computer Science, Technion---Israel Inst.\ of Technology,
     Haifa~32000, Israel \\
  Department of Computing and Information Sciences,
     Utrecht University, The Netherlands \\
  Department of Mathematics,
      California State University Northridge, Los Angeles, CA, USA
}
\received{2015-05-19}
\revised{2016-1-30}
\accepted{2016-3-7}
\publicationdetails{18}{2016}{3}{4}{646}
\begin{document}
\publicationdetails{18}{2016}{3}{4}{646}
\maketitle
\begin{abstract}
   We study the configuration space of rectangulations and convex subdivisions
   of $n$ points in the plane. It is shown that a sequence of $O(n \log n)$
   elementary {\it flip} and {\it rotate} operations can transform any
   rectangulation to any other rectangulation on the same set of $n$ points.
   This bound is the best possible for some point sets, while
   $\Theta(n)$ operations are sufficient and necessary for others.
   Some of our bounds generalize to convex subdivisions of $n$ points
   in the plane.
\end{abstract}

%------------------------------ Text -------------------------------------

\section{Introduction}

The study of rectangular subdivisions of rectangles is motivated by
VLSI floorplan design~\cite{LLY03} and cartographic visualization~\cite{EMS+12,KS07,Rai34}.
The rich combinatorial structure of rectangular floorplans has also
attracted theoretical research~\cite{BGP+08,Fel13}. Combinatorial
properties lead to efficient algorithms for the recognition and
reconstruction of the {\it rectangular graphs} induced by the corners of the
rectangles in a floorplan~\cite{HRK13,RNG04}, the {\it contact graphs} of the
rectangles~\cite{KK85,Ung53}, and the contact graphs of the horizontal and vertical
{\it line segments} that separate the rectangles~\cite{FMP95}. The number of
combinatorially different floorplans with $n$ rectangles is known to be
$B(n)=\Theta(8^n/n^4)$, the $n$th Baxter number~\cite{Yao:03}.

Rectangular subdivisions in the presence of points have also been studied in the literature.
Given a finite set $P$ of points in the interior of an axis-aligned rectangle $R$,
a {\it rectangulation} of $(R,P)$ is a subdivision of $R$ into rectangles
by pairwise noncrossing axis-parallel line segments such that every
point in $P$ lies in the relative interior of a segment (see Figure~\ref{fig:example}).
\begin{figure}[htbp]
   \centering
   \includegraphics[width=0.4\textwidth]{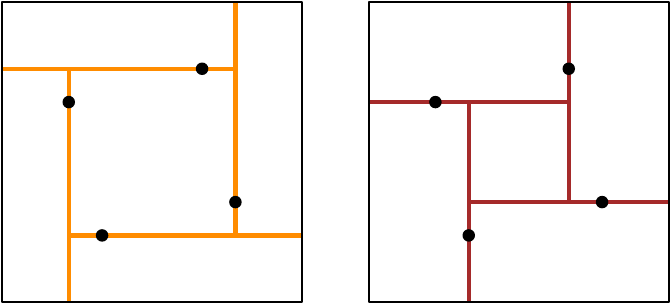}
   \caption{Two different rectangulations of a set of four points.}
   \label{fig:example}
\end{figure}

Finding a rectangulation of minimum total edge length has attracted attention~\cite{Calheiros:03,Cardei:01,Du:86,Gonzalez:Zheng:85,Gonzalez:Zheng:89,Gonzalez:Zheng:90,Lev:86,Lingas:82}
due to its applications in VLSI design and stock cutting in the presence of material defects.
This problem is known to be NP-hard~\cite{Lingas:82}, however, its complexity is unknown
when the points in $P$ are in general position in the sense that they have distinct $x$-
and $y$-coordinates, that is, the points are {\it noncorectilinear}.
It is not hard to see that in this case the minimum edge-length rectangulation must consist of exactly $n$
line segments~\cite{Calheiros:03}. For the rest of this paper, we consider only noncorectilinear point sets
$P$, and rectangulations determined by $|P|$ line segments, one containing each point in $P$.

The space of all the rectangulations of a point set $P$ (within a rectangle $R$) can be explored using the following two elementary operators introduced in~\cite{ABP06} (refer to Figure~\ref{fig:operations}).
\begin{figure}[htbp]
   \centering
   \includegraphics[width=.66\columnwidth]{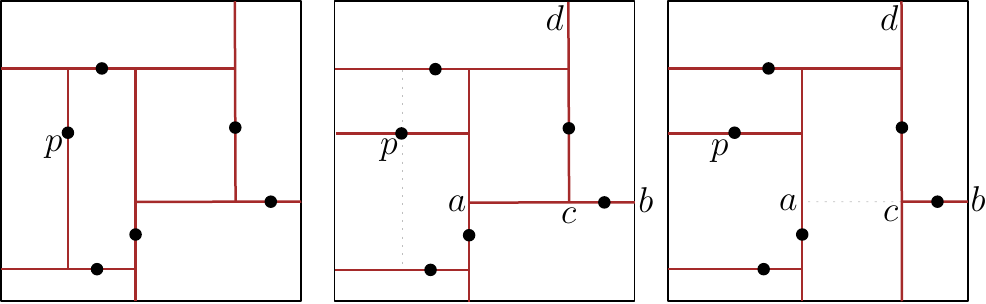}
   \caption{A rectangulation $r_1$ of a set of~6 points,
            $r_2={\mbox{\sc Flip}}(r_1,p)$, and
            $r_3=\mbox{{\sc Rotate}}(r_2,c)$.}
   \label{fig:operations}
\end{figure}

\begin{definition}\textbf{(Flip)}
   \label{def:flip}
   Let $r$ be a rectangulation of $P$ and let $p\in P$ be a point
   such that the segment $s$ that contains $p$ does not contain any endpoints
   of other segments. The operation {\it {\sc Flip}$(r,p)$} changes
   the orientation of $s$ from vertical to horizontal or vice-versa.
\end{definition}

\begin{definition}\textbf{(Rotate)}
   \label{def:rotate}
   Let $r$ be a rectangulation of $P$.
   Let $s_1=ab$ be a segment that contains $p\in P$. Let $s_2=cd$ be a segment
   such that $c$ lies in $ap\subset s_1$ and $ac$ does not contain any endpoints of other segments.
   The operation {\sc Rotate}$(r,c)$ shortens $s_1$ to $cb$, and extends $s_2$ beyond
   $c$ until it reaches another segment or the boundary of $R$.
\end{definition}

For a finite set of noncorectilinear points $P\subset R$, we denote by $G(P) = (V,E)$ the {\it graph of
rectangulations of $P$}, where the vertex set is $V = \{ r : r$ is a rectangulation of $P\}$
and the edge set is $E=\{(r_1,r_2) :$ a single flip or rotate operation on $r_1$
produces $r_2\}$. Since both operations are reversible, $G(P)$ is an undirected
graph. It is not hard to show that $G(P)$ is connected~\cite{ABP06}, and there is a sequence of $O(n^2)$ flip and rotate operations between any two rectangulations in $G(P)$ when $P$ is a set of $n$ points in $R$. It is natural to ask for the diameter of $G(P)$, which we call the {\it flip diameter} of $P$ for short. For every set of $n$ points in $R$, the flip diameter is at least $\Omega(n)$, since every point set admits a  rectangulation with all horizontal segments and one with all vertical segments, and each operation modifies at most two segments.

\subsection*{Results}

In this paper, we show that the flip diameter of $P$ is $O(n \log n)$
for every $n$-element noncorectilinear point set $P$ (Section~\ref{sec:rect}),
and it is $\Omega(n \log n)$ for some $n$-element noncorectilinear point sets
(Section~\ref{sec:construction}). However, there are $n$-elements noncorectilinear
points sets for which the flip diameter is $\Theta(n)$ (Section~\ref{sec:diagonal}).
That is, the flip diameter is always between $O(n \log n)$ and $\Omega(n)$,
depending on the point configuration, and both bounds are the best possible.

We extend the flip and rotate operations and the notion of flip diameter to convex
subdivisions (Section~\ref{sec:convex}). A {\it convex subdivision} of a set
$P\subset\mathbb{R}^2$ of points is a subdivision of the plane
into convex faces by pairwise noncrossing line segments, halflines, and lines,
each of which contains exactly one point of $P$. We show that the flip diameter
for the convex subdivisions of $n$ points is always $O(n \log n)$ and sometimes $\Theta(n)$.

\subsection*{Related work}

Determining the exact number of rectangulations on $n$ noncorectilinear points remains an
elusive open problem in enumerative combinatorics~\cite{ABP06,ABB+13}. Recently,
Felsner~\cite{Fel14} proved that every combinatorial floorplan with $n+1$ rooms
can be embedded into every set of $n$ noncorectilinear points, hence every set of $n$
noncorectilinear points has at least $B(n)=\Theta(8^n/n^4)$ rectangulations.

The currently best known upper bound, $O(18^n/n^4)$ by Ackerman~\cite{A06}, uses
the so-called ``cross-graph'' charging scheme~\cite{SS03,SW06}, originally developed
for counting the number of (geometric) triangulations on $n$ points in the plane.
This method is based on elementary ``flip'' operations that transform one
triangulation into another. Lawson~\cite{Law77} proved that every triangulation
on $n$ points in the plane in general position (i.e., no three on a line)
can be transformed into the Delaunay triangulation with $O(n^2)$ flips,
and this bound is the best possible by a construction due to Hurtado et al.~\cite{HNU99}.
However, for $n$ points in convex position, $2n{-}10$ flips are sufficient,
due to a bijection with binary trees with $n-2$ internal nodes~\cite{STT88}.
Hence the flip diameter of every triangulation on $n$ points in the plane is
always between $\Theta(n)$ and $\Theta(n^2)$ depending on the point configuration.
Eppstein et al.~\cite{EMS+12} and Buchin et al.~\cite{BEL+11} define two
elementary flip operations on floorplans, in terms of the directed dual graph,
and solve optimization problems on floorplans by traversing the flip graph.

Generalizing simultaneous edge flips in a triangulation, Meijer and Rappaport~\cite{MR04}
considered {\it simultaneous edge flips} in convex decompositions of a point set $P$.
A {\it convex decomposition} of $P$ is a subdivision of the convex hull of $P$ into convex polygons whose joint vertex set is $P$; and two convex decompositions of $P$ are related by a simultaneous edge flip if their union contains no crossing edges.

\section{An Upper Bound on the Flip Diameter of Rectangulations}
\label{sec:rect}

In this section, we show that for every set $P$ of $n$ noncorectilinear points in a rectangle $R$,
the diameter of $G(P)$ is $O(n \log n)$.

\begin{theorem}
   \label{thm:rect}
   For every noncorectilinear set $P$ of $n$ points in the plane, the diameter of the graph $G(P)$ is $O(n \log n)$.
\end{theorem}

Given a rectangulation $r$ of $P$, we construct a sequence of $O(n \log n)$ operations that
transforms $r$ into a rectangulation with all segments vertical (a {\it canonical} rectangulation). Our method relies on the concept of independent sets, defined in terms of the bar visibility graph. Let $r$ be a rectangulation of $P$. The {\it bar visibility graph}~\cite{DHLM83,TT86,Wis85} on the horizontal segments of $r$ is defined as a graph $H(r)$, where the vertices correspond to the horizontal segments in $r$; and two horizontal segments $s_1$ and $s_2$ are adjacent in $H(r)$ if and only if there are points $a\in s_1$ and $b\in s_2$ such that $ab$ is a vertical segment (not necessarily in $r$) that does not intersect any other horizontal segment in $r$. It is clear that the bar visibility graph is planar.

Observe that we can always change the orientation of any line segment $s$ with $O(n)$ operations:
successively shorten $s$ using rotate operations until $s$ contains no other segment endpoints, and then flip $s$.
This simple procedure is formulated in the following subroutine.

\begin{quote}
   \texttt{Shorten\&Flip}$(r,s)$. Let $s$ be a segment in a rectangulation $r$. Assume $s=ab$ and $p\in P$ is in the relative interior of $s$. While $s$ contains the endpoint of some other segment, let $c_1\in s$ and $c_2\in s$ be the endpoints of some other segments closest to $a$ and $b$, respectively (possibly $c_1=c_2$). If $p\not\in ac_1$, then apply {\sc Rotate}$(r,c_1)$ to shorten $s=ab$ to $c_1b$. Else, apply {\sc Rotate}$(r,c_2)$ to shorten $s$ to $ac_2$. When $s$ does not contain the endpoint of any other segment, apply {\sc Flip}$(r,p)$.
\end{quote}

The proof of Theorem~\ref{thm:rect} follows from a repeated invocation of the following lemma.

\begin{lemma}
   \label{lem:lin}
   Let $r$ be a rectangulation of a set of $n$ pairwise noncorectilinear points in a rectangle $R$.
   There is a sequence of $O(n)$ flip and rotate operations that turns at least one quarter of the
   horizontal segments into vertical segments, and keeps vertical segments vertical.
\end{lemma}

\begin{proof}
   By the four color theorem~\cite{RSST97}, $H(r)$ has an independent set $I$ that contains at least one quarter of the horizontal segments in $r$. The total number of endpoints of vertical segments that lie on some horizontal segment in $I$ is $O(n)$. Successively call the subroutine \texttt{Shorten\&Flip}$(r,s)$ for all horizontal segment $s\in I$.

   The horizontal segments in $I$ are shortened and flipped into vertical orientation. All operations maintain the invariants that (1) the segments in $I$ are pairwise disjoint, and (2) the remaining horizontal segments in $I$ form an independent set in the bar visibility graph (of all horizontal segments in the current rectangulation). It follows that each operation either decreases the number of horizontal segments in $I$ (flip), or decreases the number of segment endpoints that lie in the relative interior of a segment in $I$ (rotate). After $O(n)$ operations, all segments in $I$ become vertical. Since only the segments in $I$ are flipped (once each), all vertical segments in $r$ remain vertical, as required.
\end{proof}

\begin{proof}[of Theorem~\ref{thm:rect}.]
   Let $P$ be a set of $n$ pairwise noncorectilinear points in a rectangle $R$.
   Denote by $r_0$ the rectangulation that consists of $n$ vertical line segments,
   one passing through each point in $P$.

   We show that every rectangulation $r_1$ of $P$ can be transformed into $r_0$
   by a sequence of $O(n \log n)$ flip and rotate operations.
   By Lemma~\ref{lem:lin}, a sequence of $O(n)$ operations can decrease the
   number of horizontal segments by a factor of at least $4/3$. After at most
   $\log n/\log (4/3)$ invocations of Lemma~\ref{lem:lin}, the number of
   horizontal segments drops below~1, that is, all segments become vertical
   and we obtain $r_0$.
\end{proof}

\paragraph{Remark.} The proof of Theorem~\ref{thm:rect} is constructive, and provides an algorithm for transforming a rectangulation on $n$ noncorectilinear points to the rectangulation with $n$ vertical segments. If the rectangulations are maintained in a doubly connected edge list (DCEL) data structure, then a flip or rotate operation can be implemented in $O(1)$ time, and the bar-visibility graph can be computed in $O(n)$ time. The bottleneck of the overall running time is the 4-coloring of the bar-visibility graph. The current best algorithm for 4-coloring an $m$-vertex planar graph runs  in $O(m^2)$ time~\cite{RSST97}, and a repeated call to this algorithm to exponentially decaying bar-visibility graphs
takes $O(n^2)$ time. If we use an $O(m)$-time 5-coloring algorithm~\cite{CNS81,Fre84,Wil85}, up to $\log n/\log (5/4)$ times, the overall running time improves to $O(n)$.

\section{A Lower Bound on the Flip Diameter of Rectangulations}
\label{sec:construction}

We show that the diameter of the graph $G(P)$ is $\Omega(n \log n)$ when
$P$ is an $n$-element {\it bit-reversal point set} (alternatively,
Halton-Hammersley point set)~\cite[Section~2.2]{Cha00}.
For every integer $k\geq 0$, we define a point set $P_k$ of
size $n=2^k$ with integer coordinates lying in the square $R=[-1,n]^2$.
For an integer $m$, $0\leq m< 2^k$, with binary representation $m=\sum_{i=1}^k b_i 2^{i-1}$,
the bit-reversal gives $y(m)=\sum_{i=1}^k b_i 2^{k-i}$. The bit-reversal point set of size
$n=2^k$ is $P_k=\{ (m, y(m)): m=0,1,\ldots ,n-1\}$.
By construction, no two points in $P_k$ are corectilinear.

\begin{figure}[htbp]
   \centering
   \includegraphics[width=.7\columnwidth]{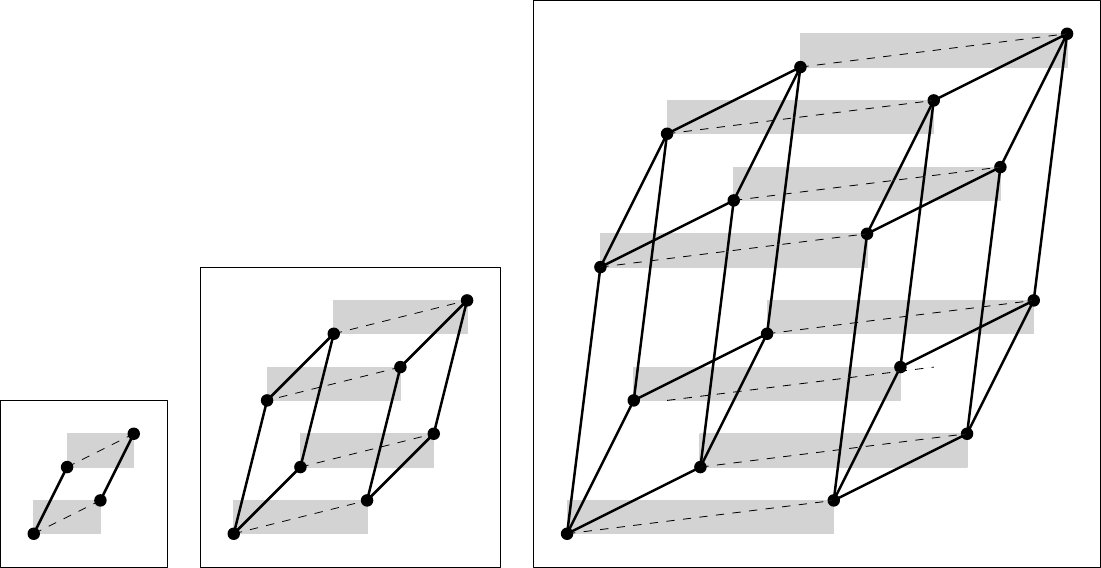}
   \caption{The sets $P_2$, $P_3$, and $P_4$. The edges connect pairs of points
   whose binary representations differ in a single bit, showing that
   $P_k$ is a projection of a $k$-dimensional hypercube.
   The grey rectangles are spanned by point pairs whose binary
   representations differ exactly in the last coordinate.}
   \label{fig:HH1}
\end{figure}

We establish a lower bound of $k2^{k-3}$ for the diameter of $G(P_k)$ using a charging
scheme. We define $k2^{k-1}$ empty rectangles (called {\it boxes}) spanned by $P_k$,
and charge one unit for ``saturating'' a box with vertical segments (as defined below).
We show that when a rectangulation with all horizontal segments is transformed into one with
all segments vertical, each box becomes saturated. We also show that each rotate (resp., flip)
operation contributes a total of at most~2 (resp.,~4) units to the saturation of various
boxes in our set. It follows that at least $(k2^{k-1})/4=k2^{k-3}=n \log n/8$ operations
are required to saturate all $k2^{k-1}$ boxes.

Consider the point set $P_k$ for some $k\in \mathbb{N}$. We say that a rectangle $B\subset [-1,n]^2$ is {\it spanned} by $P_k$ if two opposite corners of $B$ are in $P_k$; and $B$ is
{\it empty} if its interior is disjoint from $P_k$.

Let $\B$ be the set of closed rectangular boxes spanned by point pairs in $P_k$ in which
the binary representation of the $x$-coordinates $(b_1, \ldots , b_k)$ differ in precisely one bit.
See Figure~\ref{fig:HH1} for examples.
Each point in $P_k$ is incident to $k$ boxes in $\B$, since there are $k$ bits.
Every box in $\B$ is spanned by two points of $P_k$, thus
$|\B|=k\cdot |P_k|/2=k2^{k-1}$.

Each point is incident to $k$ boxes of sizes
$2^{i-1} \times 2^{k-i}$ for $i=0,\ldots , k-1$,
since changing the $i$th bit $b_i$ incurs an $2^{i-1}$ change
in the $x$-coordinate and an $2^{k-i}$ change in the $y$-coordinate.
However, if we change several bits successively, then either the
$x$-coordinate changes by more than $2^{i-1}$ or the $y$-coordinate
changes by more than $2^{k-i}$, and so every box in $\B$ is empty.
Note also that the boxes of the same size are pairwise disjoint.

We now define the ``saturation'' of each box $B\in \B$ with respect to a rectangulation of $P_k$.
Let $B\in \B$ and let $r$ be a rectangulation of $P_k$. The vertical {\it extent} of $B$ is the orthogonal projection of $B$ into the $y$-axis. Consider the vertical segments of $r$ clipped in $B$ (i.e., the segments $s\cap B$ for all vertical segments $s$ in $r$). The {\it saturation of $B$}
with respect to $r$ is the percentage of the vertical extent of $B$ covered by
projections of vertical segments of $r$ clipped in $B$. See Figure~\ref{fig:HH2}
\begin{figure}[htbp]
   \centering
   \includegraphics[width=.26\columnwidth]{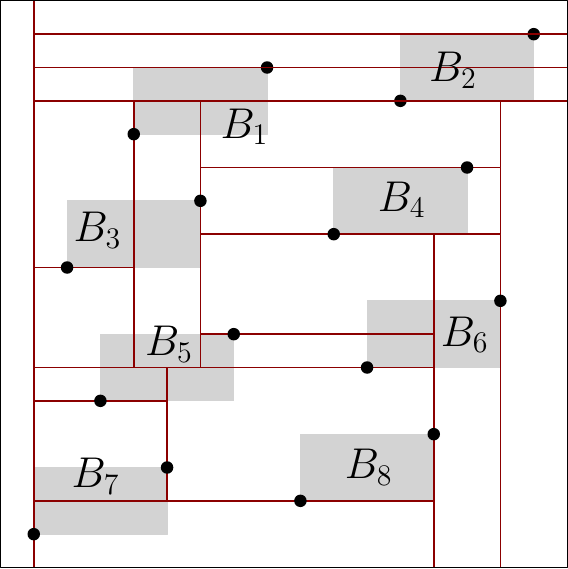}
   \caption{A rectangulation of $P_4$. The saturation of box
            $B_1, \ldots ,B_8$ is $\frac{1}{2}$, 0, 1, 0, 1, 1, 1, and 1,
            respectively.}
   \label{fig:HH2}
\end{figure}
for examples.
By definition, the saturation of $B$ is a real number in $[0,1]$. For every $B\in \B$,
we have that the saturation of $B$ is~0 when $r$ is a rectangulation with all horizontal segments, and it is~1 when $r$ consists of all segments vertical. If we transform an
all-horizontal rectangulation into an all-vertical one by a sequence of operations,
the total saturation of all $k2^{k-1}$ boxes in $\B$ increases from $0$ to $k2^{k-1}$.
The key observation is that a single operation increases
the total saturation of all boxes in $\B$ by at most a constant.

It remains to bound the impact of a single operation on the saturation of a box in $\B$. Consider first an operation {\sc Rotate}$(r,c)$ that increases the saturation of some box $B\in \B$. A rotate operation shortens a segment $s_1$ and extends an orthogonal segment $s_2$. The saturation of a box $B$ can increase only if a vertical segment grows, so we may assume that $s_1$ is horizontal and $s_2$ is vertical. Denote by $s$ the newly inserted portion of $s_2$. Note that $s$ lies in a single face of the rectangulation $r$. Similarly, if an operation {\sc Flip}$(r,p)$ increases the saturation of a box in $B$, then it replaces a horizontal segment by a vertical segment passing through $p$.
The new vertical segment lies in {\it two} adjacent faces of $r$, separated by the original horizontal segment through $p$. We represent the new vertical segment as the union of two collinear vertical segments $s\cup s'$ that meet at point $p$. In summary, an operation {\sc Rotate}$(r,c)$ inserts one vertical segment $s$ that lies in the interior of a face of $r$, and an operation {\sc Flip}$(r,p)$ inserts two such segments. We show now that if such a new vertical segment $s$ increases the saturation of some box in $B\in \B$, then $s$ must lie in $B$.

\begin{lemma}
   \label{lem:ce}
   Suppose that an operation inserts a vertical segment $s$ that lies in
   a face $f$ of the rectangulation $r$. If the insertion of $s$ increases
   the saturation of a box $B\in \B$, then $s\subset B$.
\end{lemma}

\begin{proof}
   Suppose, to the contrary, that $s\not\subset B$.
   Let $p,q\in P_k$ denote the two opposite corners of points that span $B$, such that $p$
   is the upper left or upper right corner of $B$, and $q$ is the opposite corner of $B$.
   Since $s$ increases the saturation of $B$, it must intersect $B$.
   Hence $f\cap B\neq \emptyset$. Since $s\not\subset B$, at least one of the endpoints of $s$
   lies in the exterior of $B$. Assume, without loss of generality, that the upper endpoint of
   $s$ lies outside $B$, and $p$ is the upper left corner of $B$. Then, the top side of $f$ is
   strictly above the top side of $B$. Since point $p$ cannot be in the interior of $f$,
   the left side of the face $f$ intersects the top side of $B$. Note that $s$ and the left
   side of $f$ have the same orthogonal projection on the $y$-axis. Therefore, the insertion
   of $s$ cannot increase the saturation of $B$, contradicting our assumption.
   We conclude that both endpoints of $s$ lie in $B$, and $s\subset B$.
\end{proof}

\begin{lemma}
   \label{lem:constant}
   A rotate (resp., flip) operation increases the total saturation of all boxes in $\B$ by at most~2 (resp.,~4).
\end{lemma}

\begin{proof}
   Suppose that an operation {\sc Rotate}$(r,c)$ inserts a vertical segment $s$,
   or an operation {\sc Flip}$(r,p)$ inserts two collinear vertical segments
   $s\cup s'$ that meet at $p$. By Lemma~\ref{lem:ce}, the insertion of $s$
   increases the saturation of a box $B\in \B$ of height $h$ by $|s|/h$ if
   $h\geq |s|$, and does not affect the saturation of boxes of height $h<|s|$.
   Recall that the boxes in $\B$ have only $k$ different sizes, $2^{i-1} \times 2^{k-i}$
   for $i=1,\ldots , k$, and the boxes of the same size are pairwise disjoint.
   Let $j\in \{1,2,\ldots , k\}$ be the largest index such that $|s|\leq 2^{k-j}$.
   For $i=1,\ldots , j$, segment $s$ increases the saturation of at most one box
   of height $h=2^{k-i}$, and the increase is at most $|s|/h = |s|\cdot 2^{i-k}$.
   So $s$ increases the total saturation of all boxes in $\B$ by at most
   $\sum_{i=1}^{j} |s|2^{i-k}\leq \sum_{i=1}^{j} 2^{i-j}< 2$, as required.
\end{proof}

\begin{theorem}
   \label{thm:constuction}
   For every $n\in \mathbb{N}$, there is an $n$-element point
   set $P\subset [-1,n]^2$ such that the diameter of $G(P)$ is
   $\Omega(n \log n)$.
\end{theorem}

\begin{proof}
   First assume that $n=2^k$ for some $k\in \mathbb{N}_0$. We have defined a set $P_k$ of $n=2^k$ points and a set $\B$ of $k2^{k-1}=n \log n/2$ boxes spanned by $P_k$. The total saturation of all boxes in $\B$ is~0 in the rectangulation with horizontal segments, and $|\B|=n \log n/2$ in the one with all segments vertical. By Lemma~\ref{lem:constant}, a single flip or rotate operation increases the total saturation by at most~4. Therefore, at least $n \log n/8$ operations are required to transform the horizontal rectangulation to the vertical one, and the diameter of $G(P_k)$ is at least $n \log n/8$.

   If $n$ is not a power of two, then put $k=\lfloor \log_2 n\rfloor$ and let $P\subset [-1,n]^2$ be the union of $P_k$ and $n{-}2^k$ arbitrary (noncorectilinear) points in $[2^k,n]^2$. All axis-parallel segments containing the points in $P\setminus P_k$ are in the exterior of $[-1,2^k]^2$. Therefore $k2^{k-3}=\Omega(n \log n)$ operations are required when all segments containing the points in $P_k\subset P$ change from horizontal to vertical.
\end{proof}

\section{The Flip Diameter for Diagonal Point Sets}
\label{sec:diagonal}

We say that a point set $P$ is {\it diagonal} if all points in $P$ lie on
the graph of a strictly increasing function (e.g., $f(x)=x$).
In this section we show that the flip diameter is $O(n)$ for any $n$-element
diagonal set. We denote by $\seg{p_i}$ the segment that contains $p_i$.

\begin{theorem}
   \label{thm:diagonal}
   For every $n\in \mathbb{N}$, the diameter of $G(P)$ is at most
   $11n$ when $P$ is a diagonal set of $n$ points.
\end{theorem}
\begin{proof}[of Theorem~\ref{thm:diagonal}]
   Without loss of generality, we may assume that the diagonal set is
   $P=\{p_i: i=1,\ldots,n\}$, where $p_i=(i,i)$.
   Given a rectangulation $r$ for $P$,
   we construct a sequence of at most $5.5n$ flip and rotate operations
   that transforms $r$ into a rectangulation with all segments vertical.
   Our algorithm consists of the following four phases.

   \smallskip\noindent
   \textbf{Phase~1: No three consecutive segments are parallel.}
   We describe an algorithm that, given a rectangulation $r$ of $P$,
   applies less than $3n$ operations and returns a rectangulation
   in which no three consecutive points are contained in parallel
   segments.

   \begin{quote}
      \texttt{NoThreeParallel}$(r)$.
      While there are three consecutive points in $P$ contained in parallel segments,
      let $p_{i-1},p_i,p_{i+1}\in P$ be arbitrary points such that the
      segments $\seg{p_{i-1}}$, $\seg{p_i}$, and $\seg{p_{i+1}}$ are parallel,
      and call \texttt{Shorten\&Flip}$(r,\seg{p_i})$.
   \end{quote}

   Note that \texttt{Shorten\&Flip}$(r,\seg{p_i})$ changes the orientation of the middle segment $\seg{p_i}$ only. Hence, after one invocation of \texttt{Shorten\&Flip}$(r,\seg{p_i})$, segments $\seg{p_{i-1}}$, $\seg{p_i}$, and $\seg{p_{i+1}}$ have alternating orientations, and will never be in the middle of a consecutive parallel triple. This implies that \texttt{NoThreeParallel}$(r)$ terminates after changing the orientation of at most $n/2$ segments, and no three consecutive points are in parallel segments in the resulting rectangulation.

   It remains to bound the number of flip and rotation operations: we have already seen that there are at most $n/2$ flip operations, and we shall charge each rotation operation to the segment extended by that operation. Note that after an invocation of \texttt{Shorten\&Flip}$(r,\seg{p_i})$, the two endpoints of $\seg{p_i}$ lie on $\seg{p_{i-1}}$ and $\seg{p_{i+1}}$, respectively, and so segment $\seg{p_i}$ is not extended after a flip in Phase~1. Thus, each segment can be extended in two possible directions (left/right or up/down).
   We claim that each segment is extended at most once in each direction. Suppose to the contrary that a segment $\seg{p_k}$ is extended twice in the same direction (refer to Figure~\ref{fig:no3},
   \begin{figure}[htbp]
      \centering
      \includegraphics[width=.7\columnwidth]{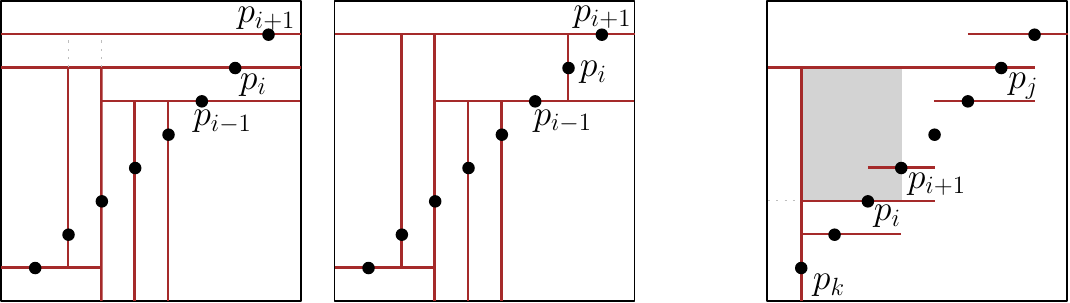}
      \caption{Left-Middle: If the segment through $p_{i-1}$, $p_i$, and
               $p_{i+1}$ are parallel, we perform
               \texttt{Shorten\&Flip}$(r_0,\seg{p_i})$.
               Right: If a rotation extends a vertical segment $\seg{p_k}$
               up from $\seg{p_i}$ to $\seg{p_j}$, then $\seg{p_{i+1}}$
               cannot be horizontal.}
      \label{fig:no3}
   \end{figure}
   right). By symmetry, we may assume that
   $\seg{p_k}$ is vertical, and it is extended upward in two rotate operations involving some horizontal segments $\seg{p_i}$ and $\seg{p_j}$ (each of which is sandwiched between two horizontal segments). In the first rotation, $\seg{p_k}$ is extended from its intersection with $\seg{p_i}$ to an intersection with $\seg{p_j}$. Hence, the order of the corresponding three points is $k<i<j$. This step creates a rectangle bounded by $\seg{p_i}$ on the top, $\seg{p_j}$ on the bottom, and $\seg{p_k}$ on the left. The right boundary of this rectangle cannot be $\seg{p_{i+1}}$, since it is horizontal; and it cannot be any segment $\seg{p_\ell}$, $\ell>i+1$, otherwise the rectangle would contain $p_{i+1}$. The contradiction confirms the claim. We conclude that this phase involves at most $2n$ rotations and at most $n/2$ flip operations.

   \smallskip\noindent
   \textbf{Phase~2: Creating a staircase.}
   We define a {\it staircase} in a rectangulation of $P$
   as a monotone increasing path along the segments that
   does not skip two or more consecutive points (refer to
   Figure~\ref{fig:staircase}).
   \begin{figure}[htbp]
      \centering
      \includegraphics[width=.42\columnwidth]{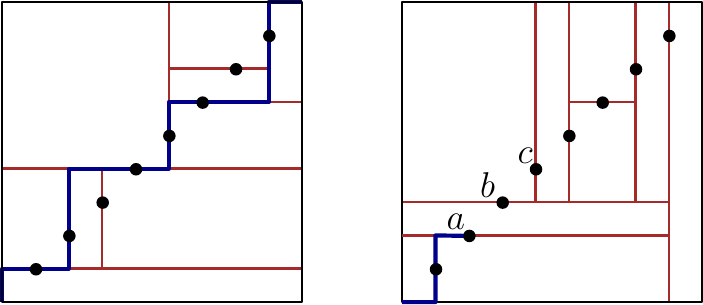}
      \caption{The left figure shows a rectangulation with a valid
               staircase, the right figure shows a rectangulation
               which does not have a valid staircase.}
      \label{fig:staircase}
   \end{figure}
   In the second phase of our algorithm,
   we transform a rectangulation $r$ with no three consecutive
   parallel segments into a rectangulation $r$ that contains a
   staircase from the lower left to the upper right corner of $R$.

   We construct a staircase $\pi$ incrementally, starting from the lower left corner.
   In each step we successively extend $\pi$ with one horizontal and one vertical
   segment to a point of $P$ or to the upper right corner of $R$. Let $\pi$ be a
   current staircase from the lower left corner to a point $a\in P$. Denote by
   $b$ and $c$ the next two points of the diagonal set $P$. If there is a monotone
   path from $a$ to $b$ or from $a$ to $c$, then append this path to $\pi$ to obtain
   a longer staircase.  Suppose there is no such path. Then the segments $\seg{a}$ and
   $\seg{b}$  must be parallel, otherwise there would be a monotone path from $a$ to $b$.
   The segment $\seg{c}$ must be perpendicular to the segment $\seg{a}$ since
   there are no three consecutive parallel segments.

   We would like to perform a rotation at the intersection of segments $\seg{b}$ and $\seg{c}$,
   and then we can extend $\pi$ from $a$ to $c$ (along $\seg{a}$ and $\seg{c}$).
   However, other operations may be necessary before we can rotate at $\seg{b}\cap \seg{c}$.
   Successively rotate all intersection points of $\seg{b}$ with other segments that
   are to the right of $b$ if $\seg{b}$ is horizontal, and that are above $b$ if $\seg{b}$
   is vertical (rotating at $\seg{b} \cap \seg{c}$ eventually).

   We claim that this phase requires at most $n$ rotations. Let $a$, $b$, and $c$ be
   defined as before, such that there is no path from $a$ to either $b$ or $c$.
   Assume, by symmetry, that $a$ and $b$ are horizontal and $c$ is vertical.
   Let $d_1, \ldots , d_k$ be segment endpoints on $\seg{b}$ to the right of $b$.
   The bottom endpoints of the vertical segments through $d_1 \hdots d_k$
   are now below $b$, and so these endpoints will never be involved in another
   rotation in our incremental procedure. Similarly, we perform a rotation at the
   left endpoint of each horizontal segment at most once.  Therefore, the total
   number of rotations is bounded by the number of segments, which is $n$.

   \smallskip\noindent
   \textbf{Phase~3: Extending all vertical segments of the staircase to full length.}
   The third phase of our algorithm returns a rectangulation in which every vertical
   segment of the staircase extends to the top and bottom sides of the bounding box $R$.
   The staircase subdivides $R$ into two regions, one above and one below
   (see Figure~\ref{fig:more_vert}).
   \begin{figure}[htbp]
      \centering
      \includegraphics[width=.63\columnwidth]{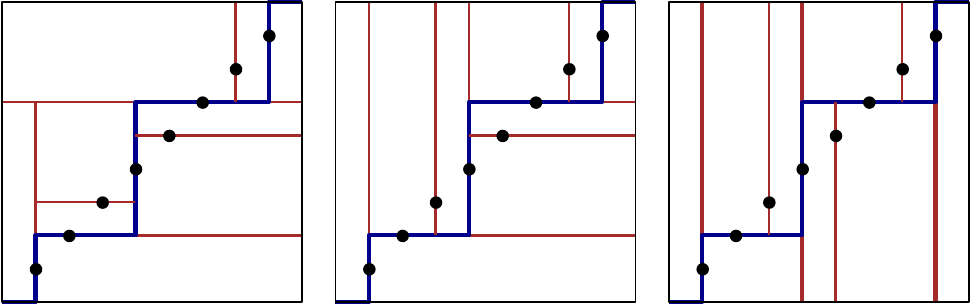}
      \caption{Left: a rectangulation with a valid staircase.
               Middle: the result of sweeping the region above the staircase.
               Right: the result of sweeping the regions on both sides of the
               staircase.}
      \label{fig:more_vert}
   \end{figure}
   We handle the two regions independently without modifying the staircase.
   To handle the region above, we sweep from top to bottom and successively shorten
   each horizontal segment we encounter until further shortening is impossible or
   would modify the staircase. If we encounter a horizontal segment $s$ that is not
   part of the staircase, we use subroutine \texttt{Shorten\&Flip}$(r,s)$ to shorten
   it as much as possible and then flip it to vertical position.
   Similarly, we sweep the region below the staircase bottom-up.

   Each flip or rotate operation in Phase~3 increases the number of segment
   endpoints lying on the top or bottom side of $R$. Therefore, the number
   of operations is bounded by the total number of segment endpoints, which is $2n$.

   \smallskip\noindent
   \textbf{Phase~4: Final steps.}
   The vertical segments of the staircase extend to the top and bottom of the
   bounding box and partition it into vertical strips (Figure~\ref{fig:more_vert}, right).
   Each strip contains a horizontal segment of the staircase, including a point in $P$.
   Because the staircase does not skip two or more consecutive points,
   there are at most three points in each strip. The total number of points
   in all strips combined is less than $n$.

   If there is just one point in a strip, then we flip the horizontal segment that contains it.
   If there are two points in a strip, then one lies on a horizontal and one on a vertical segment.
   In this case, perform a rotation so the vertical segment reaches maximum length, and then
   flip the horizontal segment. If there are three points in a strip, then the middle point
   lies on a horizontal segment and the other two lie on vertical segments above and below
   the staircase, respectively. In this case, we perform two rotations so that both vertical
   segments reach full length, and then flip the horizontal segment.
   Similarly to Phase~3, each operation in this phase increases the number of segment
   endpoints lying on the top or bottom side of $R$. Therefore, $2n$ is an upper
   bound on the number of operations in Phases~3 and~4 combined.
\end{proof}

\paragraph{Remark.}
Felsner et al.~\cite{FFN+11} showed that the rectangulations of a diagonal point set are in bijection
with {\it twin pairs of binary trees}. Visually, the part of the boundaries of the rectangles lying above
(resp., below) the points in $P$ form a binary tree whose vertices are the points in $P$ and the top-right
(resp., bottom-left) corners of $R$. This observation allows determining the number of
rectangulations for diagonal point sets in terms of Baxter permutations (cf.~\cite{FFN+11}).
Theorem~\ref{thm:diagonal} implies the following.
\begin{corollary}
For any two pairs of twin binary trees with $n$ leaves, there is a sequence of twin binary trees $(t_1,t'_1),\ldots ,(t_k,t'_k)$ of size $k=O(n)$ that starts with the first pair, ends with the second pair, and for every $i$ ($1\leq  i < k$), either $t_i=t_{i+1}$ or $t_{i+1}$ is obtained from $t_i$ by a single tree-rotate operation (the same for $t'_i$ and $t'_{i+1}$).
\end{corollary}

\section{Generalization to Convex Subdivisions}
\label{sec:convex}

Given a set $P$ of $n$ points in the plane, % $\mathbb{R}^2$,
a {\it convex subdivision} for $P$ is a subdivision of the plane
into convex cells by $n$ pairwise noncrossing line segments
(possibly lines or half-lines), such that each segment contains
exactly one point of $P$ and no three segments have a point in common.

The flip and rotate operations can be interpreted for convex subdivisions of a point set $P$, as well (see Figure~\ref{fig:operations2}).
\begin{figure}[htbp]
   \centering
   \includegraphics[width=0.7\columnwidth]{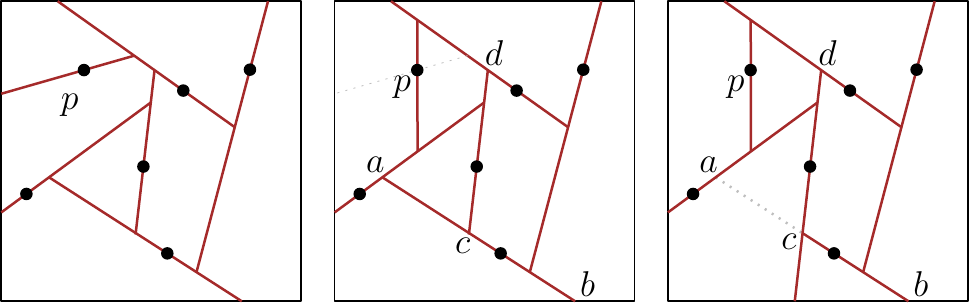}
   \caption{A convex subdivision $r_1$ of~6 points,
            $r_2=\mbox{{\sc Flip}}(r_1,p,\sigma)$,
            and $r_3=\mbox{{\sc Rotate}}(r_2,c)$.}
   \label{fig:operations2}
\end{figure}
The definition of the operation {\sc Rotate}$(r,c)$ is identical to the rectilinear version.
The operation {\sc Flip}$(r,p)$ requires more attention, since a segment may have infinitely many possible orientations.

\begin{definition}[\textbf{Flip}]
   \label{def:flip+}
   Let $r$ be a convex subdivision of $P$, let $p\in P$ be a point such
   that the segment $s$ containing $p$ does not contain any endpoints
   of other segments, and let $\sigma\in \mathbb{S}^1$ be a unit vector.
   The operation {\it {\sc Flip}$(r,p,\sigma)$} replaces $s$ by a segment
   of direction $\sigma$ containing $p$.
\end{definition}

Similarly to the graph of rectangulations $G(P)$, we define the {\it graph of convex subdivisions of $P$}, $\widehat{G}(P) = (V,E)$, where the vertex set is $V = \{ r : r$ is a convex subdivision  of $P\}$ and the edge set is $E=\{(r_1,r_2) :$ a single flip or rotate operation on $r_1$ produces $r_2\}$. Our main result in this section is that even though $\widehat{G}(P)$ is an infinite graph, its diameter is $O(n \log n)$, where $n=|V|$.

\begin{theorem}
   \label{thm:subd}
   For set $P$ of $n$ points, the graph $\widehat{G}(P)$ is connected
   and its diameter is $O(n \log n)$.
\end{theorem}

We show that any convex subdivision can be transformed into a subdivision with all segments vertical through a sequence of $O(n \log n)$ operations. Subroutine \texttt{Shorten\&Flip}$(r,s)$ from Section~\ref{sec:rect} can be adapted almost verbatim: for a unit vector $\sigma\in \mathbb{S}^1$, subroutine \texttt{Shorten\&Flip}$(r,s,\sigma)$ shortens segment $s$ maximally
by rotate operations, and then flips it to direction $\sigma$.

\begin{lemma}
   \label{lem:dir}
   Let $r$ be a convex subdivision of a set of $n$ points in the plane with distinct $x$-coordinates.
   There is a sequence of $O(n)$ flip and rotate operations that turns at least a $\frac{1}{36}$ fraction
   of the nonvertical segments vertical, and keeps all vertical segments vertical.
\end{lemma}

Before proving Lemma~\ref{lem:dir}, we need to introduce a few technical terms. Consider a convex subdivision $r$ of a set of $n$ points with distinct $x$-coordinates.  We say that a segment $s_1$ {\it hits} another segment $s_2$ if an endpoint of $s_1$ lies in the relative interior of $s_2$. An {\it extension of $s_1$  beyond $s_2$ hits $s_3$} if $s_1$ hits $s_2$ and $s_1$ is contained in a segment $s_1'$, such that $s_1'$ hits $s_3$ and $s_1'$ crosses at most one segment (namely, $s_2$). Note that the operation {\sc Rotate}$(r,s_1\cap s_2)$, if applicable,
would extend $s_1$ beyond $s_2$ to hit segment $s_3$. We define the {\it extension visibility digraph} $\widehat{H}(r)$ on all segments in $r$, where the vertices correspond to the segments in $r$, and we have a directed edge $(s_2,s_3)$ if $s_2$ hits $s_3$ or there is a segment $s_1$ such that an extension of $s_1$ beyond $s_2$ hits $s_3$.

\begin{figure}[htbp]
   \centering
   \includegraphics[width=0.25\columnwidth]{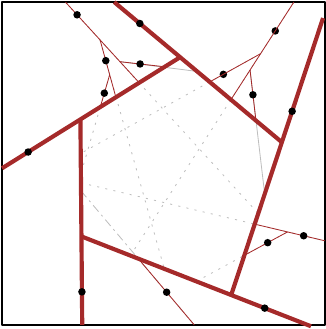}
   \caption{A convex subdivision $r$ of a set of~15 points. The five bold
            segments induce $K_5$ in the extension visibility graph
            $\widehat{H}(r)$.}
   \label{fig:k5}
\end{figure}
The graph $\widehat{H}(r)$ is not necessarily planar: it is not difficult to construct a convex subdivision $r$ for a set of $O(t^2)$ points where $\widehat{H}(r)$ contains a subgraph isomorphic to the complete graph $K_t$ (Figure~\ref{fig:k5}). It is enough to describe the segments of such a construction (the points may lie anywhere in the interior of the segments). Start with the sides of a convex $t$-gon, and extend the sides in counterclockwise direction to obtain a convex subdivision. The extension visibility digraph of these $t$ segments contains a cycle $(s_1,\ldots, s_t)$. Add short segments in the exterior of the initial $t$-gon: If a segment hits $s_i$ and its supporting line passes though $s_j$, it generates an edge $(s_i,s_j)$ in the extension visibility digraph.

Note that the number of edges in $\widehat{H}(r)$ is at most $4n$, since each segment hits at most two other segments, but some segments extend to infinity; and the extension of each segment beyond each of its endpoints hits at most one other segment. If $P$ contains $n$ points, the average degree in $\widehat{H}(r)$ is less than $8$. Therefore, $\widehat{H}(r)$ has an independent set of size at least $n/9$ (obtained by successively choosing minimum-degree vertices~\cite{Hoc83,Tur41}).

\begin{proof}[of Lemma~\ref{lem:dir}.]
   Let $r$ be a convex subdivision of a set of $n$ points with distinct $x$-coordinates. Let $I_0$ be an independent set in the extension visibility graph $\widehat{H}(r)$ induced by all nonvertical segments. As noted above, $I_0$ contains at least $1/9$ of the nonvertical segments in $r$. Let $I_1\subseteq I_0$ be an independent set in the bar visibility graph of the segments in $I_0$ (two nonvertical segments in $I_0$ are mutually visible if there is a
   vertical segment between them that does not cross any segment of the subdivision). Since the bar visibility graph is planar, we have $|I_1|\geq |I_0|/4$, and so $I_1$ contains at least a $1/36$ fraction of the nonvertical segments in $r$. The total number of segment endpoints that lie in the relative interior of segments in $I_1$ is $O(n)$. An invocation of subroutine \texttt{Shorten\&Flip}$(r,s,(0,1))$ for each segment $s\in I_1$ changes their orientation to vertical.

   The operations maintain the invariants that
   (1)~the segments in $I_1$ are pairwise disjoint; and
   (2)~the nonvertical segments in $I_1$ form an independent set in both $\widehat{H}$ and the bar visibility graph of all segments in $I_0$.
   It follows that each operation either decreases the number of nonvertical segments in $I_1$ (flip), or decreases the number of segment endpoints that lie in the relative interior of a nonvertical segment in $I_1$ (rotate). After performing $O(n)$ operations, all segments in $I_1$ become vertical. Since only the segments in $I_1$ change orientation (each of them is flipped to become vertical), all vertical segments in $r$ remain vertical, as required.
\end{proof}

\begin{proof}[of Theorem~\ref{thm:subd}.]
   Let $P$ be a set of $n$ points in a bounding box. We may assume, by rotating the point set if necessary, that the points in $P$ have distinct $x$-coordinates. Denote by $r_0$ the convex subdivision given by $n$ vertical line segments, one passing through each point in $P$.

   Consider a convex subdivision $r_1$ of $P$. By Lemma~\ref{lem:dir}, $O(n)$ operations can decrease the number of nonhorizontal segments by a factor of at least $36/35$.  After at most $\log n/\log (36/35)$ invocations of Lemma~\ref{lem:dir}, the number of nonvertical segments drops below~1, that is, all segments become vertical and we obtain $r_0$, as claimed.
\end{proof}

\subsection*{A linear upper bound for collinear points}

We show that the bound $O(n \log n)$ on the diameter of the flip graph
$\widehat{G}(P)$ from Theorem~\ref{thm:subd} can be improved to $O(n)$ for
some simple point configurations.

\begin{theorem}
   \label{thm:collinear}
   For every $n\in \mathbb{N}$, the diameter of $\widehat{G}(P)$ is
   $O(n)$ when $P$ is a set of $n$ collinear points.
\end{theorem}

\begin{proof}
   We may assume that $P=\{p_i: i=1,\ldots , n\}$, $p_i=(i,0)$. Let $r$ be
   a convex subdivision of $P$. No segment in $r$ is horizontal, since
   each segment contains a unique point. We show that there is a sequence of
   $O(n)$ operations that transforms $r$ into a convex subdivision with all
   segments vertical. Suppose that not all segments in $r$ are vertical, and
   let $m>0$ be the minimum absolute value of the slope of a nonvertical
   segment. Our algorithm proceeds in two phases.

   \smallskip\noindent
   \textbf{Phase~1: Building a staircase.}
   In this phase, we transform $r$ into a convex subdivision in which
   every $\seg{p_i}$ is a nonvertical ray with a left endpoint at infinity,
   and has slope $m/n$ when $i$ is odd and $-m/n$ when $i$ is even.
   We process the points $p_1,\ldots, p_n$ successively in this order.
   We maintain the following invariant (refer to Figure~\ref{fig:collinear1}):
   \begin{figure}[htbp]
      \centering
      \includegraphics[width=0.95\columnwidth]{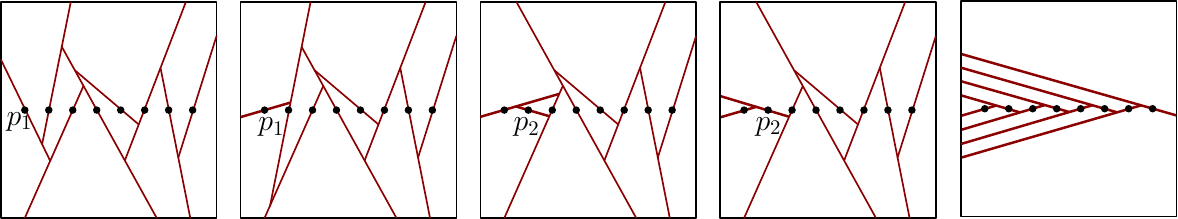}
      \caption{Transforming a convex subdivision for $n$ points lying on the
               $x$-axis into rays of alternating slopes $m/n$ and $-m/n$.}
      \label{fig:collinear1}
   \end{figure}
   When we start processing $p_i$, segments $\seg{p_1},\ldots,\seg{p_{i-1}}$
   are already rays with the desired properties, and they do not contain
   the endpoints of any segment $\seg{p_i},\ldots,\seg{p_n}$ in their
   relative interiors.

   Suppose that points $p_j$, for all $1\leq j<i$, have already been processed.
   We use the subroutine \texttt{Shorten\&Flip}$(r,\seg{p_i},\sigma)$ to shorten $\seg{p_i}$
   and flip it into the desired orientation (of slope $m/n$ or $-m/n$ depending on
   the parity of $i$). For $i=1$, this already guarantees that segment $\seg{p_1}$
   is a ray with a left endpoint at infinity. If $i>1$, however, the left endpoint
   of $\seg{p_i}$ lies on the ray $\seg{p_{i-1}}$, whose slope is different
   from that of $\seg{p_i}$. Due to the invariant, $\seg{p_{i-1}}$ contains
   no segment endpoints to the right of the intersection point $\seg{p_{i-1}} \cap \seg{p_i}$.
   We can now apply a rotation at $\seg{p_{i-1}} \cap \seg{p_i}$. Since the slopes of all
   segments are at least $m/n$ in absolute value, $\seg{p}_i$ extends to infinity,
   and our invariant is established for segments $\seg{p_1},\ldots ,\seg{p_i}$.

   We argue that Phase~1 uses $O(n)$ operations. For each point $p_i$, we invoke
   \texttt{Shorten\&Flip}$(r,\seg{p_i},\sigma)$, which performs only one flip operation,
   followed by at most one rotation.
   It is enough to bound the number of rotations performed in the invocations of subroutines
   \texttt{Shorten\&Flip}. We claim that these invocations extend each segment at
   most once in each direction. Suppose, to the contrary, that a segment $\seg{p_k}$
   is extended in the same direction twice (when processing points $p_i$ and $p_j$).
   Since the points are processed from left to right, we have $i<j<k$.
   Consider the step in which a rotation extends $\seg{p_k}$ from its intersection
   with $\seg{p_i}$ to an intersection with $\seg{p_j}$. Then, the $x$-axis and segment
   $\seg{p_k}$ intersect $\seg{p_i}$ and $\seg{p_j}$ in different orders.
   Hence, $\seg{p_i}$ and $\seg{p_j}$ must cross each other, contradicting
   the fact that they are part of a convex subdivision. This proves our
   claim. It follows that Phase~1 uses at most $5n$ operations.

   \smallskip\noindent
   \textbf{Phase~2: Orienting all lines to be vertical.}
   Refer to Figure~\ref{fig:collinear2}.
   \begin{figure}[htbp]
      \centering
      \includegraphics[width=0.95\columnwidth]{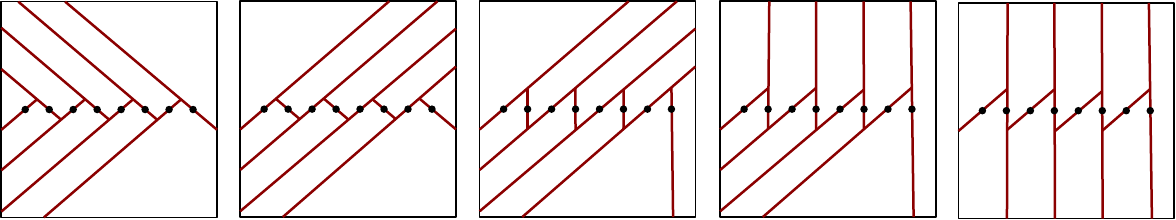}
      \caption{The first four passes over the rays yield vertical lines
               through all even points, separating the odd points.}
      \label{fig:collinear2}
   \end{figure}
   We are given a convex subdivision in which every $\seg{p_i}$ is a ray with a left endpoint at infinity, and has slope $m/n$ when $i$ is odd and $-m/n$ when $i$ is even. Note that only consecutive segments intersect. (This is still not a ``canonical'' subdivision, as $m$ depends on the initial convex subdivision.)
   In Phase~2, we transform all segments into vertical lines.
   We make five passes over the odd or even points:
   (1)~For every odd point $p_i$ from right to left, rotate ray $\seg{p_i}$ into a line of slope $m/n$.
   As a result, the lines through the odd points separate the even points.
   (2)~We can now flip the segments through the even points independently into vertical segments.
   (3)~For every even point $p_i$ from left to right, rotate the top endpoints to infinity.
   (4)~For every even point from right to left, rotate the bottom endpoint to infinity. We obtain vertical lines through the even points, that separate the odd points.
   (5)~Finally, flip the segments through the odd points independently into vertical lines. We made three passes over even segments and two passes over odd segments, so the total number of operations in Phase~2 is $2n+\lceil n/2\rceil\leq 3n$. The two phases combined use at most $8n$ operations.
\end{proof}

\section{Conclusions}

We have shown that the diameter of the flip graph $G(P)$ is between  $\Omega(n)$ and $O(n \log n)$ for every $n$-element point set $P$, and these bounds cannot be improved. The diameter is $\Theta(n)$ for diagonal point sets, and $\Theta(n \log n)$ for the bit-reversal point set. The flip graph $G(P)$ of a noncorectilinear set $P$ is uniquely determined by the permutation of the $x$- and $y$-coordinates of the points~\cite{ABP06} (e.g., diagonal point sets correspond to the identity permutation). It is an open problem to find the average diameter of $G(P)$ over all $n$-element permutations. It would already be of interest to find broader families of point sets with linear diameter:
Is the diameter of $G(P)$ linear if $P$ is in convex position, unimodal, or corresponds to a separable permutation (see~\cite{BBL98})?

We have shown that the diameter of the flip graph is also % bounded by
$O(n \log n)$ for the convex subdivisions of $n$ points in the plane. We do not know whether this bound is tight. It is possible that the flip diameter is $\Omega(n \log n)$ for the bit-reversal point set defined in Section~\ref{sec:construction}, but our proof of Theorem~\ref{thm:constuction} heavily relies on axis-aligned boxes and does not seem to extend to convex subdivisions.

Given a convex subdivision $r$ of a point set $P$, the flip and rotate operations can be thought of as a continuous deformation; refer to Figure~\ref{fig:operations3}:
\begin{figure}[htbp]
   \centering
   \includegraphics[width=.7\columnwidth]{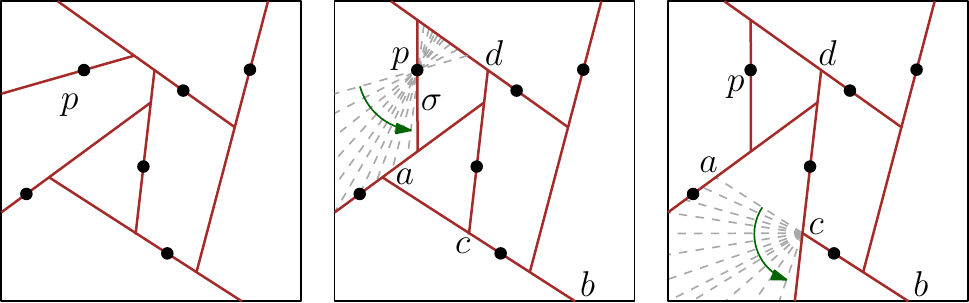}
   \caption{Left: a convex subdivision $r_1$ of~6 points.
            Middle: operation $\mbox{{\sc Flip}}(r_1,p,\sigma)$ performed
            as a continuous deformation of the segment containing $p$.
            Right: operation $\mbox{{\sc Rotate}}(r_2,c)$  performed as a
            continuous deformation of part of the segment containing $c$.}
   \label{fig:operations3}
\end{figure}
{\sc Flip}$(r,p,\sigma)$ rotates the segment containing $p$ continuously to position $\sigma$; and {\sc Rotate}$(r,c)$ rotates continuously a portion of the segment containing $c$ into the extension of the segment that currently ends at $c$. The {\it weight} of an operation can be defined as the number of vertices swept during this continuous deformation. By Theorem~\ref{thm:subd}, a sequence of $O(n \log n)$ operations can transform any convex subdivision to any other convex subdivision on $n$ points. A single operation, however, may have $\Omega(n)$ weight. We conjecture that the weighted diameter of the graph $\widehat{G}(P)$ is also $O(n \log n)$ for every $n$-elements point set $P$.

%\acknowledgments
\paragraph{Acknowledgments.}
Research on this paper was initiated at the Workshop on
{\it Counting and Enumerating of Plane Graphs} held in March~2013 at
Schloss Dagstuhl.
We thank Sonia Chauhan, Michael Hoffmann, and Andr\'{e} Schulz for
insightful comments and stimulating conversations about the problems
discussed in this paper.

\nocite{*}
\bibliographystyle{abbrvnat}
% use the following instead if you encounter problems
%\bibliographystyle{alpha}

%\nocite{*}
%\bibliographystyle{abbrvnat}
%%% use the following instead if you encounter problems
%%%\bibliographystyle{alpha}
%\bibliography{sample-dmtcs}
%\label{sec:biblio}

\end{document}